\documentclass[smallextended,12pt,psfig,a4]{amsart}
\usepackage{latexsym}
\usepackage{amsmath}
\usepackage{amsfonts}
\usepackage{amssymb}
\usepackage{bm}
\usepackage{geometry}
\usepackage{subcaption}
\usepackage{caption}
\usepackage{multirow, bigdelim}
\usepackage{graphicx}

\usepackage{hyperref}

\usepackage{enumitem}

\usepackage[dvipsnames,table,xcdraw]{xcolor}

\DeclareMathOperator{\sgn}{sgn}

\newtheorem{theorem}{Theorem}[section]

\newtheorem{pro}[theorem]{Proposition}
\newtheorem*{conj*}{Conjecture}

\newtheorem{remark}[theorem]{Remark}

\theoremstyle{definition}

\newtheorem{example}[theorem]{Example}

\theoremstyle{remark}

\numberwithin{equation}{section}

\newcommand{\opn}[1]{\operatorname{#1}}

\newcommand{\bs}[1]{\boldsymbol{#1}}

\def\eps{\varepsilon}
\def\>{\rangle}
\def\<{\langle}

\def\sgn{\opn{sgn}}

\def\0{\bs{0}}
\def\1{\mathbbm{1}}

  \def\XXint#1#2#3{{\setbox0=\hbox{$#1{#2#3}{\int}$}
      \vcenter{\hbox{$#2#3$}}\kern-.47\wd0}}

\usepackage{xcolor,soul}

\begin{document}
\pagestyle{plain}

\def\beq{\begin{equation}}
\def\eeq{\end{equation}}
\def\eps{\epsilon}
\def\laa{\langle}
\def\raa{\rangle}
\def\qed{\begin{flushright} $\square$ \end{flushright}}
\def\qee{\begin{flushright} $\Diamond$ \end{flushright}}
\def\ov{\overline}
\def\bma{\begin{bmatrix}}
\def\ema{\end{bmatrix}}

\def\ora{\overrightarrow}

\def\bma{\begin{bmatrix}}
\def\ema{\end{bmatrix}}
\def\bex{\begin{example}}
\def\eex{\end{example}}
\def\beq{\begin{equation}}
\def\eeq{\end{equation}}
\def\eps{\epsilon}
\def\laa{\langle}
\def\raa{\rangle}
\def\qed{\begin{flushright} $\square$ \end{flushright}}
\def\qee{\begin{flushright} $\Diamond$ \end{flushright}}
\def\ov{\overline}

\author{Manuel D. de la Iglesia}
\address{Manuel D. de la Iglesia, Universidad de Alcal\'a, Departamento de F\'isica y Matem\'aticas, Campus universitario, E-28805 Alcal\'a de Henares (Madrid), Spain.}
\email{manuel.dominguezi@uah.es}

\author{Carlos F. Lardizabal}
\address{Carlos F. Lardizabal, Instituto de Matem\'atica e Estat\'istica, Universidade Federal do Rio Grande do Sul. Porto Alegre, RS  91509-900 Brazil.}
\email{cfelipe@mat.ufrgs.br}

\thanks{The work of the first author is supported by the Spanish Agencia Estatal de Investigaci\'on research project [PID2024-155133NB-I00], entitled \emph{Ortogonalidad, aproximaci\'on e integrabilidad: aplicaciones en procesos estoc\'asticos cl\'asicos y cu\'anticos}.}

\date{\today}


\title{Exact solutions of open quantum Brownian motions \\on the real line for two-level systems}

\begin{abstract}
We investigate open quantum Brownian motions as quantum analogues of classical diffusion processes under interaction with an external enviroment. Building upon the microscopic derivation by Sinayskiy and Petruccione \cite{sp15}, we revisit the associated master equation and study its formulation as a generalized parabolic system. Employing Fourier transform methods, we derive exact analytical solutions for one-dimensional evolutions of particles with two-level internal degree of freedom. 
\end{abstract}

\maketitle



\medskip

\section{Introduction}

Random walks \cite{durrett,grimmett} are pervasive in our understanding of physical phenomena, with extensive applications ranging from biology \cite{berg} and chemistry \cite{vkampen} to economics \cite{malkiel} and computer science \cite{motwani}. In recent years, quantum generalizations of random walks have been actively explored, particularly due to their potential applications in quantum computation \cite{nielsen,portugal,salvador}, information theory and cryptography \cite{schwest,vla}. These developments have also brought renewed interest in models that capture the dynamics of quantum systems interacting with their environment, a domain known as open quantum dynamics \cite{benatti,BP}.

\medskip

A class of examples of interest within this context consist of the open quantum Brownian motions (OQBMs), which serves as a quantum analogue of the classical diffusion process under the influence of an external environment. This model not only provides insight into fundamental questions of decoherence and dissipation, but also offers a rich mathematical structure that connects quantum theory, probability, and functional analysis. In \cite{BBT2} a detailed description of OQBM is described in terms of open quantum walks \cite{attal} and associated continuous time limits, also see \cite{BBT1}.  In \cite{sp15}, Petruccione and Sinayskiy have provided a derivation of an OQBM for a free quantum Brownian particle with two
degrees of freedom, with the corresponding master equation for the reduced density $\rho_s(t)$ obtained in terms of a Born-Markov approximation, namely,
\beq\label{eq1sp}
\frac{d}{dt}\rho_s(t)=-\int_0^\infty \mathrm{Tr}_B\Bigg(H_{SB}(t),[H_{SB}(t-\tau),\rho_s(t)\otimes\rho_B]\Bigg)\;d\tau,\eeq
where $H_{SB}$ denote the system-bath interaction hamiltonian. In such work, numerical integrations for particular cases are studied, but a general analytical solution was not pursued. In the present article, we are motivated by this latter work in order to derive exact analytical solutions for systems on the line using Fourier transform methods. The basic setting is obtained from noting that the differential equations derived from (\ref{eq1sp}) can be written as a generalized parabolic system \cite{volpert},
\begin{equation*}\
\frac{\partial}{\partial t} \vec{u}(t,x) = 
2\gamma_p \, \frac{\partial^2}{\partial x^2} \vec{u}(t,x)
+ B\frac{\partial}{\partial x} \vec{u}(t,x)+ C \, \vec{u}(t,x),\quad (t,x) \in [0,\infty) \times \mathbb{R},
\end{equation*}
for certain choices of matrices $B$ and $C$. We review the deduction of the master equation appearing in \cite{sp15}, and the corresponding details leading to the analytical solutions are presented in the following sections.

\medskip

We remark that in the present work we will focus on 1-dimensional evolutions, namely, motions on the real line with the corresponding differential equations describing an initial value problem, and we consider a quantum particle described by  a point in the Bloch sphere \cite{nielsen} (i.e., an order 2 density matrix corresponding to a particle with two-level internal degree of freedom) together with its location on the line. It is our intention, in a future publication, to consider boundary value problems, so that we are able to describe motions on the half-integer line and on finite intervals as well.

\medskip

The contents of this work are as follows. In Section \ref{sec2} we make a brief review of the microscopic derivation of OQBM as presented in \cite{sp15} and we provide basic definitions and settings for subsequent calculations. In Section \ref{sec3} the general approach is presented, noting that further specific expressions can be obtained for certain choices of parameters. These appear in Sections \ref{sec4}, \ref{sec5} and \ref{sec6}. For convenience of the reader, a brief Appendix revises the Born-Markov approximation and other preliminaries.

\bigskip

\section{Brief review: microscopic derivation of the OQBM}\label{sec2}

This description follows \cite{sp15} and, for convenience, provides a brief revision on the derivation of the master equation for the OQBM. In Subsection \ref{rewrite}, we write the specific system of equations to be studied in this work. The Hamiltonian of the quantum Brownian particle is defined to be
$$
H_S=\frac{P^2}{2M}+\frac{\omega_0}{2}\sigma_z+\Omega\sigma_x,$$
where the first term is the Hamiltonian of the free Brownian particle, the second term is the Hamiltonian of the free two-level system, and the last term $\Omega\sigma_x$ describes a {\bf weak classical driving of the system}, where $\Omega\ll \omega_0$. Here, as usual, $\sigma_x$ and $\sigma_z$ denote the standard Pauli matrices.

\medskip

The environment Hamiltonian is given by
$$H_B=\sum_n\omega_na_n^*a_n,$$
where $a_n^*$ and $a_n$ are bosonic creation and annihilation operator with standard commutation relations $([a_n,a_m^*]=\delta_{m,n})$ and $\omega_n$ are the frequencies of the corresponding oscillators. The interaction Hamiltonian is
$$H_{SB}=\sum_n iP\xi_n \alpha\frac{a_n^*-a_n}{\sqrt{2}}+\beta\sigma_z\xi_n\frac{a_n^*+a_n}{\sqrt{2}},$$
where the real constants $\alpha$ and $\beta$ describe the strength of the decoherence in external and internal degrees of freedom and the coefficients $\xi_n$ describe the strength of the coupling to the environment. In \cite{sp15}, it is shown that
$$H_{SB}(t)=J^*B(t)+JB^*(t),$$
where
$$J=i\alpha P+\beta\sigma_z,\;\;\;B(t)=\sum_n \frac{\xi_n}{\sqrt{2}}a_ne^{-i\omega_n t}$$
and from this one can obtain an equation for $\rho_s$, namely
\beq\label{be1}\frac{d}{dt}\rho_s(t)=\gamma(0_-)\mathcal{D}(J)\rho_s(t)+\gamma(0_+)\mathcal{D}(J^*)\rho_s(t),\eeq
where
$$\mathcal{D}(M)\rho=M\rho M^*-\frac{1}{2}\{M^*M,\rho\}_+,\;\;\;\gamma(\omega)=\int_{-\infty}^{+\infty}e^{i\omega s}\mathrm{Tr}_B[B^*(s)B(0)]\;ds.$$
Regarding practical calculations, {the amplitude $\Omega$ of the weak external driving term is of the same order as the decoherence coefficients appearing in equation (\ref{be1}) for $\rho_s$, that is, $\Omega \sim \gamma(0_{\pm})\ll \omega_0$}.

\medskip

Finally, in order to include the effect of external driving into the master equation (\ref{be1}) we rotate such equation with the unitary operator
$$U_\Omega=\exp[-it\Omega\sigma_x],$$
so that we obtain
\begin{equation}\label{eq:122}
\frac{d}{dt}\rho_s(t)=-i[\Omega\sigma_x,\rho_s]+\gamma_p\mathcal{D}(P)\rho_s+\gamma_z\mathcal{D}(\sigma_z)\rho_s+i\Delta(P\rho_s\sigma_z-\sigma_z\rho_s P),
\end{equation}
where
$$\gamma_p=\alpha^2(\gamma(0_+)+\gamma(0_-)),\;\;\;\gamma_z=\beta^2(\gamma(0_+)+\gamma(0_-)),\;\;\;\Delta=\alpha\beta(\gamma(0_-)-\gamma(0_+)).$$

For a generic OQW the density matrix of the quantum walker has a diagonal form in the position space
\begin{equation*}
    \rho(t) = \sum_n \rho_n \otimes |n\rangle\langle n|.
\end{equation*}
Accordingly, in the OQBM case the density matrix will be given by \cite{BBT1},
\begin{equation}\label{eq:12}
    \rho_s(t) = \int_{-\infty}^{\infty} dx \, \rho(t,x) \otimes |x\rangle\langle x|,
\end{equation}
with $P(t,x) = \operatorname{Tr}[\rho(t,x)]$ the probability density of finding the system at position $x$ at time $t$, where the trace is taken over the internal degree of freedom of the open quantum Brownian particle. By direct substitution of the density matrix (\ref{eq:12}) into the master equation \eqref{eq:122} we obtain the following equation,
\begin{equation}\label{eq:13}
\begin{split}
    \frac{\partial}{\partial t}\rho(t,x) = & 2\gamma_p \frac{\partial^2}{\partial x^2} \rho(t,x)- i[\Omega\sigma_x, \rho(t,x)]+ \gamma_z(\sigma_z \rho(t,x) \sigma_z - \rho(t,x))\\
    &\qquad- \Delta \left( \sigma_z \frac{\partial \rho(t,x)}{\partial x} + \frac{\partial \rho(t,x)}{\partial x} \sigma_z \right).
    \end{split}
\end{equation}
The above master equation (\ref{eq:13}) describes OQBM. This master equation has the same structure as the master equation introduced by Bauer et al (\cite[equation (2)]{BBT1} and \cite[equation (28)]{BBT2}). The propagation of the Brownian particle is described by the diffusive term $2\gamma_p \frac{\partial^2 \rho(t,x)}{\partial x^2}$. The dissipative dynamics of the internal degree of freedom of the Brownian particle is described by the Lindblad term $-i[\Omega\sigma_x, \rho(t,x)] + \gamma_c(\sigma_z \rho(t,x) \sigma_z - \rho(t,x))$. The last term $-\Delta\sigma_z\frac{\partial \rho(t,x)}{\partial x} - \Delta\frac{\partial \rho(t,x)}{\partial x}\sigma_z$ of the master equation (\ref{eq:13}) is a ``decision making'' term and describes the environment mediated interaction between the external and internal degrees of freedom of the quantum Brownian particle. The presence of this term makes the quantum Brownian motion ``open'' and this term plays the role of a ``quantum coin'', which affects the direction of the propagation of the Brownian particle.

\medskip

For a density matrix of the form 
$$
\rho(t,x)=\begin{pmatrix}\rho_{11}(t,x)&\rho_{12}(t,x)\\
\ov{\rho_{12}(t,x)}&\rho_{22}(t,x)
\end{pmatrix},
$$
where $t\geq0,x\in\mathbb{R}$, the system of coupled linear partial differential equations given by \eqref{eq:13} (see \cite{sp15}), which is an instance of a diffusion-advection-reaction system \cite{DAR2,DAR}, can be written as

\begin{equation}\label{e1}
\begin{split}
\frac{\partial}{\partial t}\rho_{11}(t,x)&=2\gamma_p\frac{\partial^2}{\partial x^2}\rho_{11}(t,x)-i\Omega\left[\rho_{21}(t,x)-\rho_{12}(t,x)\right]-2\Delta\frac{\partial}{\partial x}\rho_{11}(t,x),\\
\frac{\partial}{\partial t}\rho_{12}(t,x)&=2\gamma_p\frac{\partial^2}{\partial x^2}\rho_{12}(t,x)-i\Omega\left[\rho_{22}(t,x)-\rho_{11}(t,x)\right]-2\gamma_z\rho_{12}(t,x),\\
\frac{\partial}{\partial t}\rho_{21}(t,x)&=2\gamma_p\frac{\partial^2}{\partial x^2}\rho_{21}(t,x)-i\Omega\left[\rho_{11}(t,x)-\rho_{22}(t,x)\right]-2\gamma_z\rho_{21}(t,x),\\
\frac{\partial}{\partial t}\rho_{22}(t,x)&=2\gamma_p\frac{\partial^2}{\partial x^2}\rho_{22}(t,x)-i\Omega\left[\rho_{12}(t,x)-\rho_{21}(t,x)\right]+2\Delta\frac{\partial}{\partial x}\rho_{22}(t,x),
\end{split}
\end{equation}
where $\gamma_p, \gamma_z, \Delta, \Omega$ are positive constants. Note that the equation for $\rho_{21}(t,x)=\ov{\rho_{12}(t,x)}$ is just the conjugate of the second equation. Assume also that an initial condition is fixed, given by some density matrix 
\begin{equation}\label{iccs}
\rho(0,x)=\begin{pmatrix}\psi_{11}(x)&\psi_{12}(x)\\
\ov{\psi_{12}(x)}&\psi_{22}(x)
\end{pmatrix},\quad x\in\mathbb{R}.
\end{equation}

%
%

\subsection{Rewriting the master equation}\label{rewrite} 

In \cite{sp15}, the authors also consider the system of equations (\ref{e1}) in terms of $\rho_{\pm}(t,x)=\rho_{11}(t,x)\pm\rho_{22}(t,x)$, $C_R(t,x)=\Re(\rho_{12}(t,x))$ and $C_I(t,x)=\Im(\rho_{12}(t,x))$. Therefore, we may write
\begin{equation}\label{e2}
\begin{split}
\frac{\partial}{\partial t}\rho_+(t,x)&=2\gamma_p\frac{\partial^2}{\partial x^2}\rho_+(t,x)-2\Delta\frac{\partial}{\partial x}\rho_-(t,x),\\
\frac{\partial}{\partial t}C_R(t,x)&=2\gamma_p\frac{\partial^2}{\partial x^2}C_R(t,x)-2\gamma_zC_R(t,x),\\
\frac{\partial}{\partial t}C_I(t,x)&=2\gamma_p\frac{\partial^2}{\partial x^2}C_I(t,x)-2\gamma_z C_I(t,x)+\Omega\rho_-(t,x),\\
\frac{\partial}{\partial t}\rho_-(t,x)&=2\gamma_p\frac{\partial^2}{\partial x^2}\rho_-(t,x)-2\Delta\frac{\partial}{\partial x}\rho_+(t,x)-4\Omega C_I(t,x).\\
\end{split}
\end{equation}
The initial conditions are changed accordingly, so that we have
\begin{equation}\label{iccs2}
\rho_\pm(0,x)=\psi_{11}(x)\pm\psi_{22}(x),\quad C_R(0,x)=\Re(\psi_{12}(x)),\quad C_I(0,x)=\Im(\psi_{12}(x)).
\end{equation}
While the authors of \cite{sp15} explored this system through numerical integration for particular cases, a general analytical solution was not pursued. In this work, we address this by deriving exact analytical solutions for the system \eqref{e2} with initial conditions \eqref{iccs2} using Fourier transform methods.

\medskip

We note that the quantities $\rho_{\pm}(t,x)=\rho_{11}(t,x)\pm\rho_{22}(t,x)$, $C_R(t,x)=\Re(\rho_{12}(t,x))$ and $C_I(t,x)=\Im(\rho_{12}(t,x))$ are closely related to the usual Bloch vector correspondence for one particle with two-level internal degree of freedom, namely,
$$
\rho(t,x)=\begin{pmatrix} \rho_{11}(t,x) & \rho_{12}(t,x) \\ \ov{\rho_{12}}(t,x) & \rho_{22}(t,x)
\end{pmatrix}\;\Longleftrightarrow\; \vec{r}(t,x)=\left(2\Re \left(\rho_{12}(t,x)\right),-2\Im\left(\rho_{12}(t,x)\right),\rho_{11}(t,x)-\rho_{22}(t,x)\right).
$$
In physical terms,  $\rho_-(t,x)$ corresponds to a {\bf population imbalance} and equals $\langle\sigma_z\rangle_x$, the local expectation value of the Pauli $\sigma_z$ operator at position $x$. Also, $C_R(t,x)$ equals $\tfrac{1}{2}\,\langle\sigma_x\rangle_x$, measuring \textbf{phase coherence} along the $x$-axis of the Bloch sphere, and $C_I(t,x)$ equals $-\tfrac{1}{2}\,\langle\sigma_y\rangle_x$, measuring phase coherence along the $y$-axis.

\medskip

Regarding the solution of the above system, first we observe that the second equation in \eqref{e2} is uncoupled and can be solved explicitly. Indeed, this is a time modification of the heat equation, whose Green's function is
\[
g(t,x) = \frac{e^{-2 \gamma_z t}}{\sqrt{8 \pi \gamma_p t}}
\exp\left(-\frac{x^2}{8 \gamma_p t}\right), \quad t>0,\quad x\in\mathbb{R}.
\]
Hence the solution for general initial data $C_R(0,x)$ is the convolution with \(G\):
\[
C_R(t,x) = \int_{-\infty}^{\infty} g(t,x-y)C_R(0,y)dy
=\frac{e^{-2 \gamma_z t}}{\sqrt{8 \pi \gamma_p t}}
\int_{-\infty}^{\infty} \exp\left(-\frac{(x-y)^2}{8 \gamma_p t}\right)\Re(\psi_{12}(y))dy.
\]
With one equation removed from the system \eqref{e2}, the remaining three can be expressed in matrix form as the following generalized parabolic system:
\begin{equation}\label{GPS}
\frac{\partial}{\partial t} \vec{u}(t,x) = 
2\gamma_p \, \frac{\partial^2}{\partial x^2} \vec{u}(t,x)
+ B\frac{\partial}{\partial x} \vec{u}(t,x)+ C \, \vec{u}(t,x),\quad (t,x) \in [0,\infty) \times \mathbb{R},
\end{equation}
where $\vec{u}(t,x) =\left(\rho_+(t,x),C_I(t,x),\rho_-(t,x)
\right)^T$ and
\begin{equation}\label{BBCC}
B=\begin{pmatrix}
0 & 0 & -2\Delta \\
0 & 0 & 0 \\
-2\Delta & 0 & 0
\end{pmatrix},\quad
C =
\begin{pmatrix}
0 & 0 & 0 \\
0 & -2\gamma_z & \Omega \\
0 & -4\Omega & 0
\end{pmatrix}.
\end{equation}
The initial conditions (see \eqref{iccs2}) are given by
\begin{equation}\label{iccs3}
\vec{u}_0(x)=\vec{u}(0,x)=\left(\psi_{11}(x)+\psi_{22}(x),\Im(\psi_{12}(x)),\psi_{11}(x)-\psi_{22}(x)\right)^T.
\end{equation}

\medskip

\section{General approach}\label{sec3}

The generalized parabolic system \eqref{GPS} is solvable for any matrices $B, C \in \mathbb{C}^{m \times m}$, provided that the components of the initial condition $\vec{u}_0(x)$ are $L^1(\mathbb{R})$ functions. Indeed, take the Fourier transform in $x$:
\[
\widehat{\vec{u}}(t,\xi) = \int_{\mathbb{R}} \vec{u}(t,x) e^{-i\xi x} dx.
\]
Then derivatives become multiplications: 
\(\partial_x \mapsto i\xi, \; \partial_x^2 \mapsto -\xi^2\), 
and \eqref{GPS} becomes an ODE in $\xi$:
\[
\frac{\partial}{\partial t} \widehat{\vec{u}}(t,\xi) = Q(\xi) \, \widehat{\vec{u}}(t,\xi), 
\]
where
\begin{equation*}\label{eq:Qz}
Q(\xi) = -2\gamma_p \xi^2 I_m + i \xi B + C,
\end{equation*}
The solution in Fourier space is
\[
\widehat{\vec{u}}(t,\xi) = \exp(t Q(\xi)) \, \widehat{\vec{u}}_0(\xi).
\]
Applying the inverse Fourier transform gives the formal solution
\begin{equation}\label{eq:solPDEg}
\vec{u}(t,x) = \frac{1}{2\pi} \int_{\mathbb{R}} e^{i \xi x} \, \exp(t Q(\xi)) \, \widehat{\vec{u}}_0(\xi) \, d\xi
= \int_{\mathbb{R}} G(t,x-y) \, \vec{u}_0(y) \, dy,
\end{equation}
where $G(t,x)$ is the matrix-valued Green's function
\begin{equation}\label{genGreen}
G(t,x)=\frac{1}{2\pi}\int_{\mathbb{R}}\exp\left(t Q(\xi)\right) e^{i \xi x}d\xi.
\end{equation}
This solution satisfies \eqref{GPS} because differentiation in $x$ corresponds to multiplication by $i\xi$ in Fourier space. Computing the Green's function $G(t,x)$ is difficult because the matrix exponential $\exp(tQ(\xi))$ is not straightforward to find. In our case, the matrices $B$ and $C$ generally do not commute (except when $\Omega=0$ or $\Delta=0$), which complicates things.

\smallskip

If we assume that $Q(\xi)$ admits a diagonal decomposition of the form
\begin{equation*}\label{eq:diagonalization}
Q(\xi) = U(\xi) \Lambda(\xi) U^{-1}(\xi),
\end{equation*}
with
\[
\Lambda(\xi) = \operatorname{diag}(\lambda_1(\xi), \ldots, \lambda_m(\xi)), \quad \text{and} \quad \Re(\lambda_k(\xi)) \leq 0 \ \text{for all } k = 1,\ldots,m,
\]
then, a formal solution can be written as \eqref{eq:solPDEg} where $G(t,x)$ is the matrix-valued Green's function for the problem \eqref{GPS} and can be defined formally as
\begin{equation}\label{genGreenF}
G(t,x)=\frac{1}{2\pi}\int_{\mathbb{R}}U(\xi) \exp\left(t \Lambda(\xi)\right)U^{-1}(\xi) e^{i \xi x}d\xi.
\end{equation}
In our setting, using matrices $B$ and $C$ in \eqref{BBCC}, we have that
\begin{equation}\label{QQQ}
Q(\xi)=\begin{pmatrix}
    -2\gamma_p\xi^2 & 0 & -2i\xi\Delta \\
    0 & -2\gamma_p\xi^2 - 2\gamma_z & \Omega \\
    -2i\xi\Delta & -4\Omega & -2\gamma_p\xi^2
\end{pmatrix}.
\end{equation}
The eigenvalues of $Q(\xi)$ can be computed explicitly. Define the functions:
\begin{align*}
p(\xi) &= 12 \Delta^6 \xi^6 + 12 \Delta^4 (3\Omega^2 + 2\gamma_z^2) \xi^4 + 12 \Delta^2 (3\Omega^4 - 5\Omega^2 \gamma_z^2 + \gamma_z^4) \xi^2+ 3 \Omega^4 (4\Omega^2 - \gamma_z^2), \\
q(\xi) &= 4 \gamma_z (-18 \Delta^2 \xi^2 + 9 \Omega^2 - 2 \gamma_z^2), \\
r(\xi) &= 4 \Delta^2 \xi^2 + 4 \Omega^2 - \tfrac{4}{3} \gamma_z^2.
\end{align*}
Then, the eigenvalues of $Q(\xi)$ are given by:
\begin{equation}\label{eigval}
\begin{split}
\lambda_1(\xi) &= -2 \gamma_p \xi^2 - \tfrac{2}{3} \gamma_z
+ \tfrac{1}{3} \left( q(\xi) + 12 \sqrt{p(\xi)} \right)^{1/3}
- \frac{r(\xi)}{\left( q(\xi) + 12 \sqrt{p(\xi)} \right)^{1/3}}, \\
\lambda_2(\xi) &= -2 \gamma_p \xi^2 - \tfrac{2}{3} \gamma_z
- \tfrac{1}{6} \left( q(\xi) + 12 \sqrt{p(\xi)} \right)^{1/3}
+ \frac{r(\xi)}{2 \left( q(\xi) + 12 \sqrt{p(\xi)} \right)^{1/3}} \\
&\quad + \frac{i \sqrt{3}}{2} \left( \tfrac{1}{3} \left( q(\xi) + 12 \sqrt{p(\xi)} \right)^{1/3}
+ \frac{r(\xi)}{\left( q(\xi) + 12 \sqrt{p(\xi)} \right)^{1/3}} \right), \\
\lambda_3(\xi) &= -2 \gamma_p \xi^2 - \tfrac{2}{3} \gamma_z
- \tfrac{1}{6} \left( q(\xi) + 12 \sqrt{p(\xi)} \right)^{1/3}
+ \frac{r(\xi)}{2 \left( q(\xi) + 12 \sqrt{p(\xi)} \right)^{1/3}} \\
&\quad - \frac{i \sqrt{3}}{2} \left( \tfrac{1}{3} \left( q(\xi) + 12 \sqrt{p(\xi)} \right)^{1/3}
+ \frac{r(\xi)}{\left( q(\xi) + 12 \sqrt{p(\xi)} \right)^{1/3}} \right).
\end{split}
\end{equation}
The matrix $U(\xi)$ that diagonalizes $Q(\xi)$ can be obtained explicitly, noting that the eigenvectors $v_{i}(\xi), i=1,2,3,$ can be expressed as
$$
v_{i}(\xi)=C(\xi)\begin{pmatrix}
-\dfrac{2i\Delta\xi}{\,2\gamma_p\xi^2+\lambda_i(\xi)\,} \\[1.2ex]
\dfrac{\Omega}{\,2\gamma_p\xi^2+2\gamma_z+\lambda_i(\xi)\,} \\[1.2ex]
1
\end{pmatrix},\quad i=1,2,3,
$$
for some (possibly normalization) factor $C(\xi)$. If one of the denominators vanish, then we can choose a different free component (other than the third one) and solve directly the system $(\lambda I-Q(\xi))v=0$.

\medskip

For the Green's function \eqref{genGreenF} to be stable as $t \to \infty$, the eigenvalues $\lambda_i(\xi)$, $i=1,2,3$, must satisfy $\Re \lambda_i(\xi) \le 0$. This is the content of the next result.

\begin{pro}
Let $\Delta,\gamma_z,\gamma_p,\Omega>0$ and $\xi\in\mathbb{R}$.  Let 
$\lambda_1(\xi),\lambda_2(\xi),\lambda_3(\xi)$ be the eigenvalues \eqref{eigval} of $Q(\xi)$ in \eqref{QQQ}. Then:
\begin{itemize}
\item If $\xi\neq0$ then $\Re\lambda_j(\xi)<0$ for $j=1,2,3$.
\item If $\xi=0$ one has $\lambda_1(0)=0$ and $\Re\lambda_{j}(0)<0$ for $j=2,3$.
\end{itemize}

\end{pro}
\begin{proof}
The characteristic polynomial of $Q(\xi)$ is
\[
\chi(\lambda)=\mbox{det}(\lambda I-Q(\xi))=\lambda^{3}+a_{1}\lambda^{2}+a_{2}\lambda+a_{3},
\]
where
\[
\begin{aligned}
a_{1} &= 6\gamma_p\xi^{2}+2\gamma_z,\\[4pt]
a_{2} &=12\gamma_p^{2}\xi^{4}+(4\Delta^{2}+8\gamma_p\gamma_z)\xi^{2}+4\Omega^{2},\\[4pt]
a_{3} &=\xi^2(8\gamma_p^{3}\xi^4 + (8\Delta^{2}\gamma_p+8\gamma_p^{2}\gamma_z)\xi^2+8\Omega^{2}\gamma_p+8\Delta^{2}\gamma_z).
\end{aligned}
\]
For $\xi\neq0$ we have $a_1,a_2,a_3>0$. Moreover a direct algebraic simplification shows
\[
a_1a_2-a_3=64\gamma_p^3\xi^6+(16 \Delta^2\gamma_p+64 \gamma_p^2 \gamma_z)\xi^4+(16 \Omega^2\gamma_p+16\gamma_p \gamma_z^2)\xi^2+8\Omega^2 \gamma_z.
\]
Hence $a_1a_2-a_3>0$. By the Routh--Hurwitz criterion for cubic polynomials \cite{Gantmacher1959}, all roots of $\chi(\xi)$ then have strictly negative real parts. For $\xi=0$, one of the eigenvalues of $Q(0)$ is 0, while the other two are $-\gamma_z\pm\sqrt{\gamma_z^2-4\Omega^2}$. Since $\gamma_z>\sqrt{\gamma_z^2-4\Omega^2}$ we have that the real parts are strictly negative.
\end{proof}

The complex structure of the eigenvalues \eqref{eigval} of $Q(\xi)$ makes it unlikely to obtain an explicit expression for the inverse Fourier transform in the Green's function \eqref{genGreenF}. Consequently, in order to derive explicit solutions, it is necessary to set some of the parameters $\Omega$, $\Delta$, or $\gamma_z$ to zero (the case $\gamma_p=0$ is excluded, as it removes the diffusion term from the process). This will be the goal of the following sections.

\section{Explicit solutions for $\Omega=0$}\label{sec4}

Setting $\Omega=0$ in the master equation \eqref{eq:13} effectively removes the unitary contribution in the Lindblad term, 
thereby eliminating the coherent part of the dissipative dynamics associated with the internal degree of freedom of the Brownian particle. In this situation the expression of the eigenvalues and eigenvectors simplify considerably. Indeed, the eigenvalues are given by
\begin{align*}
\lambda_1(\xi) &= -2\gamma_p \, \xi^2 - 2 \gamma_z, \\
\lambda_2(\xi) &= -2\gamma_p \, \xi^2 + 2 i \Delta \, \xi, \\
\lambda_3(\xi) &= -2\gamma_p \, \xi^2 - 2 i \Delta \, \xi,
\end{align*}
while the corresponding diagonalization matrix $U(\xi) = U$ is constant:
\[
U = \begin{pmatrix}
 0& -1 &1  \\
 1& 0 &0  \\
 0& 1 &1 
\end{pmatrix}.
\]
From here we can easily compute $e^{t Q(\xi)}$, which is given by
\[
e^{t Q(\xi)} = 
e^{-2 \gamma_p t \, \xi^2} 
\begin{pmatrix}
\cos(2 \Delta t \, \xi) & 0 & -i \, \sin(2 \Delta t \, \xi) \\
0 & e^{-2 \gamma_z t} & 0 \\
- i \, \sin(2 \Delta t \, \xi) & 0 & \cos(2 \Delta t \, \xi)
\end{pmatrix}.
\]
By performing the inverse Fourier transform to $e^{t Q(\xi)}$ we obtain the
matrix-valued Green's function \eqref{genGreen}, given by
\begin{equation*}\label{ggffo0}
G(t,x) =
\frac{e^{-x^2/(8\gamma_p t)}}{\sqrt{8 \pi \gamma_p t}}
\begin{pmatrix}
e^{-\Delta^2 t/(2\gamma_p)} \cosh\left(\frac{x \Delta}{2\gamma_p}\right) & 0 & e^{-\Delta^2 t/(2\gamma_p)} \sinh\left(\frac{x \Delta}{2\gamma_p}\right) \\
0 & e^{-2\gamma_z t} & 0 \\
e^{-\Delta^2 t/(2\gamma_p)} \sinh\left(\frac{x \Delta}{2\gamma_p}\right) & 0 & e^{-\Delta^2 t/(2\gamma_p)} \cosh\left(\frac{x \Delta}{2\gamma_p}\right)
\end{pmatrix}.
\end{equation*}

\medskip

We are now ready to derive explicit solutions of the generalized parabolic system \eqref{GPS}, provided that suitable initial conditions \eqref{iccs} are prescribed. Observe that, given the structure of the matrix-valued Green's function above, 
there is no advantage in considering a non-diagonal initial density matrix $\rho(0,x)$, since all the information on the probability density and the population imbalance is contained in the main diagonal.

\begin{enumerate}
\item\textbf{Gaussian initial condition}. Consider a diagonal initial condition with different Gaussian distributions of the form:
\begin{equation}\label{iccsGau}
\rho(0,x)=\begin{pmatrix}\frac{p}{\sqrt{2\pi}\,\sigma_1} \, e^{-x^2/(2\sigma_1^2)}&0\\
0&\frac{1-p}{\sqrt{2\pi}\,\sigma_2} \, e^{-x^2/(2\sigma_2^2)}
\end{pmatrix},\; x\in\mathbb{R},\; 0<p<1,\;\sigma_1,\sigma_2>0.
\end{equation}
Then (see \eqref{iccs3})
$$
\vec{u}_0(x)=\left(\frac{p}{\sqrt{2\pi}\,\sigma_1} \, e^{-x^2/(2\sigma_1^2)}+\frac{1-p}{\sqrt{2\pi}\,\sigma_2} \, e^{-x^2/(2\sigma_2^2)},0,\frac{p}{\sqrt{2\pi}\,\sigma_1} \,e^{-x^2/(2\sigma_1^2)}-\frac{1-p}{\sqrt{2\pi}\,\sigma_2} \, e^{-x^2/(2\sigma_2^2)}\right)^T.
$$ 
The explicit solution of \eqref{GPS} can be computed directly from \eqref{eq:solPDEg} (since the convolution of Gaussian distributions is another Gaussian distribution), in which case we have
\[
\vec{u}(t,x)=\frac{1}{\sqrt{2\pi}\sqrt{4\gamma_z t + \sigma_1^2}}\begin{pmatrix}
p\,e^{-\frac{(x-2\Delta t )^2}{8\gamma_z t + 2 \sigma_1^2}}+(1-p)e^{-\frac{(x+2\Delta t)^2}{8\gamma_z t + 2 \sigma_2^2}} \\[2ex]
0\\[2ex]
p\,e^{-\frac{(x-2\Delta t )^2}{8\gamma_z t + 2 \sigma_1^2}}-(1-p)e^{-\frac{(x+2\Delta t)^2}{8\gamma_z t + 2 \sigma_2^2}}
\end{pmatrix},
\]
and the solutions of the density matrix $\rho(t,x)$ are given by $\rho_{12}(t,x)=0$ and
\[
\rho_{11}(t,x)=
\frac{p\,e^{-\frac{(x-2\Delta t )^2}{8\gamma_z t + 2 \sigma_1^2}}}{\sqrt{2\pi} \,\sqrt{4\gamma_z t + \sigma_1^2}},\quad\rho_{22}(t,x)=\frac{(1-p)e^{-\frac{(x+2\Delta t)^2}{8\gamma_z t + 2 \sigma_2^2}}}{\sqrt{2\pi} \, \sqrt{4\gamma_z t + \sigma_2^2}}.
\]
Therefore, the probability density $P(t,x) = \operatorname{Tr}[\rho(t,x)]$ is given by
\begin{equation}\label{PDF1}
P(t,x)=\frac{p\,e^{-\frac{(x-2\Delta t )^2}{8\gamma_z t + 2 \sigma_1^2}}}{\sqrt{2\pi} \,\sqrt{4\gamma_z t + \sigma_1^2}}+\frac{(1-p)e^{-\frac{(x+2\Delta t)^2}{8\gamma_z t + 2 \sigma_2^2}}}{\sqrt{2\pi} \, \sqrt{4\gamma_z t + \sigma_2^2}},
\end{equation}
while the population imbalance $Q(t,x)=\rho_{11}(t,x)-\rho_{22}(t,x)$ is
\begin{equation}\label{PI11}
Q(t,x)=\frac{p\,e^{-\frac{(x-2\Delta t )^2}{8\gamma_z t + 2 \sigma_1^2}}}{\sqrt{2\pi} \,\sqrt{4\gamma_z t + \sigma_1^2}}-\frac{(1-p)e^{-\frac{(x+2\Delta t)^2}{8\gamma_z t + 2 \sigma_2^2}}}{\sqrt{2\pi} \, \sqrt{4\gamma_z t + \sigma_2^2}}.
\end{equation}
Figure \ref{fig1} shows the probability distribution and the population imbalance of the open quantum Brownian particle for different values of the parameters and time evolution. Observe that a particular case of this example (for $p=3/4$ and $\sigma_1=\sigma_2=\sqrt{2}/2$) was already studied in (17) of \cite{sp15}.
\begin{figure}[!ht]
    \centering
    \begin{subfigure}[b]{0.5\textwidth}
        \centering
        \includegraphics[height=3.50in]{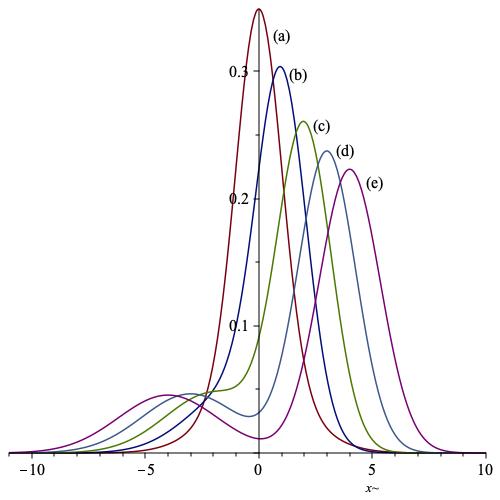}
    \end{subfigure}%
    ~ 
    \begin{subfigure}[b]{0.5\textwidth}
        \centering
        \includegraphics[height=3.50in]{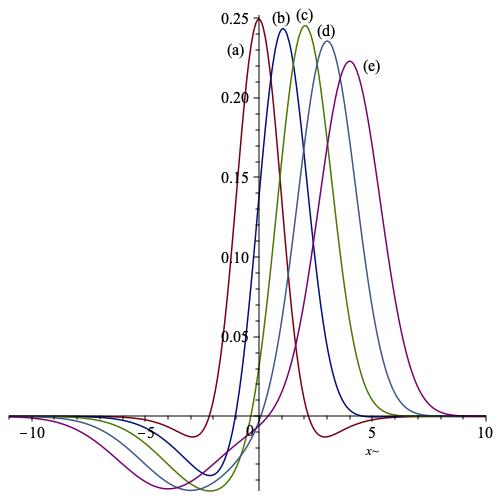}
           \end{subfigure}
  \caption{The probability distribution \eqref{PDF1} (left) and the population imbalance \eqref{PI11} (right) of the OQBM for different moments of time. The initial distribution is given by \eqref{iccsGau}. The curves from (a) to (e) corresponds to times $0, 50, 100, 150, 200$, respectively. The other parameters are chosen to be $p=3/4,\sigma_1=1,\sigma_2=2,\gamma_p=10^{-3},\Delta=10^{-2}, \gamma_z=10^{-3}$.}
\label{fig1}
\end{figure}

\medskip

\item \textbf{Laplacian initial condition}. Consider now an initial condition with different Laplace distributions of the form
\begin{equation}\label{iccsGau2}
\rho(0,x)=\begin{pmatrix}\displaystyle\frac{p}{2a}e^{-|x|/a}&0\\
0&\displaystyle\frac{1-p}{2b}e^{-|x|/b}
\end{pmatrix},\; x\in\mathbb{R},\; a,b>0,\; 0<p<1.
\end{equation}
The explicit solution of \eqref{GPS} can be computed directly from \eqref{eq:solPDEg}, in which case we have
\[
\vec{u}(t,x)=\begin{pmatrix}\rho_{11}(t,x)+\rho_{2}(t,x) \\[2ex]
0 \\[2ex]
\rho_{11}(t,x)-\rho_{2}(t,x) 
\end{pmatrix},
\]
where
\[
\rho_{11}(t,x)=\frac{p}{4a} \, e^{\frac{2\gamma_p t}{a^{2}}}
\left[
e^{\frac{2\Delta t - x}{a}} \,
\operatorname{erfc}\!\left(
\frac{t\Delta - \frac{x}{2} + \frac{2\gamma_p t}{a}}{\sqrt{2\gamma_p t}}
\right)
+
e^{-\frac{2\Delta t - x}{a}} \,
\operatorname{erfc}\!\left(
\frac{-t\Delta + \frac{x}{2} + \frac{2\gamma_p t}{a}}{\sqrt{2\gamma_p t}}
\right)
\right],
\]
\[
\rho_{22}(t,x)=\frac{1-p}{4b}\,e^{\frac{2\gamma_p t}{b^{2}}}\!
\left[
e^{\frac{2\Delta t + x}{b}}\,
\operatorname{erfc}\!\!\left(\frac{t\Delta + \tfrac{x}{2} + \tfrac{2\gamma_p t}{b}}
{\sqrt{2\gamma_p t}}\right)
+
e^{-\frac{2\Delta t + x}{b}}\,
\operatorname{erfc}\!\!\left(\frac{-t\Delta - \tfrac{x}{2} + \tfrac{2\gamma_p t}{b}}
{\sqrt{2\gamma_p t}}\right)
\right].
\]
Here $\operatorname{erfc}(x)$ is the \emph{complementary error function}, given by $\operatorname{erfc}(x)=1-\operatorname{erf}(x)$ where $\operatorname{erf}(x)=\frac{2}{\sqrt{\pi}}\int_0^xe^{-t^2}dt$ is the usual \emph{error function}.

\medskip

The probability density and the population imbalance are given by 
\begin{equation}\label{PDF2}
P(t,x) = \operatorname{Tr}[\rho(t,x)]=\rho_{11}(t,x)+\rho_{22}(t,x),\quad Q(t,x) =\rho_{11}(t,x)-\rho_{22}(t,x).
\end{equation}
Figure \ref{fig2} shows the probability distribution and the population imbalance of the open quantum Brownian particle for different values of the parameters and time evolution.

\begin{figure}[!ht]
    \centering
    \begin{subfigure}[b]{0.5\textwidth}
        \centering
        \includegraphics[height=3.50in]{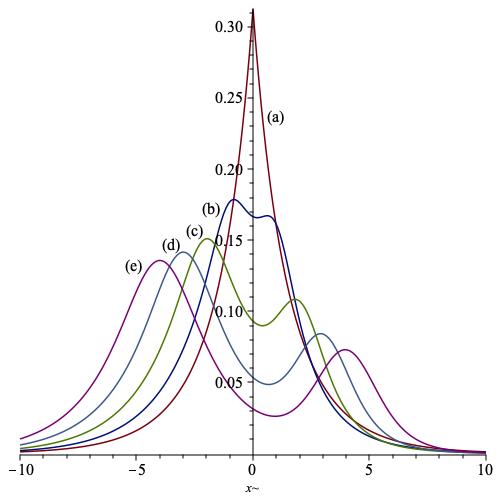}
    \end{subfigure}%
    ~ 
    \begin{subfigure}[b]{0.5\textwidth}
        \centering
        \includegraphics[height=3.50in]{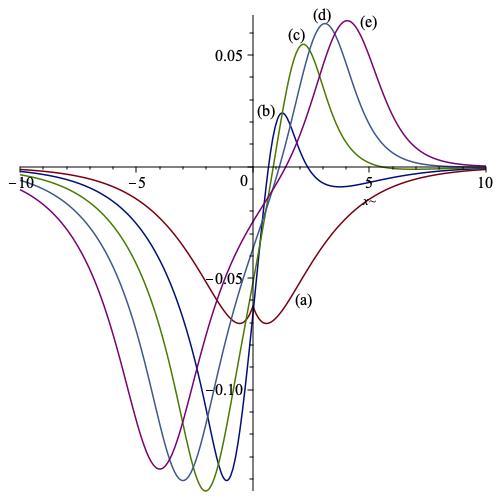}
           \end{subfigure}
  \caption{The probability distribution \eqref{PDF2} (left) and the population imbalance \eqref{PDF2} (right) of the OQBM for different moments of time. The initial distribution is given by \eqref{iccsGau2}. The curves from (a) to (e) corresponds to times $0, 50, 100, 150, 200$, respectively. The other parameters are chosen to be $p=1/4,a=1,b=2,\gamma_p=10^{-3},\Delta=10^{-2}, \gamma_z=10^{-3}$.}
\label{fig2}
\end{figure}

%

\medskip

\item \textbf{Uniform initial distribution}. Consider now an initial condition with different uniform distributions centered at $x=0$:
\begin{equation}\label{iccsGau3}
\rho(0,x)=\begin{pmatrix}\displaystyle\frac{p}{2a}\chi_{[-a,a]}(x)&0\\
0&\displaystyle\frac{1-p}{2b}\chi_{[-b,b]}(x)
\end{pmatrix},\; x\in\mathbb{R},\; a,b>0,\; 0<p<1,
\end{equation}
where $\chi_A$ denotes the indicator function. The explicit solution of \eqref{GPS} can be computed directly from \eqref{eq:solPDEg}, in which case we have
\[
\vec{u}(t,x)=\begin{pmatrix}\rho_{11}(t,x)+\rho_{2}(t,x) \\[2ex]
0 \\[2ex]
\rho_{11}(t,x)-\rho_{2}(t,x) 
\end{pmatrix},
\]
where
\[
\rho_{11}(t,x)=\frac{p}{4a} \,
\left[
\operatorname{erf}\!\left(
\frac{x+a-2t\Delta}{2\sqrt{2\gamma_p t}}
\right)
-
\operatorname{erf}\!\left(
\frac{x-a-2t\Delta}{2\sqrt{2\gamma_p t}}
\right)
\right],
\]
\[
\rho_{22}(t,x)=\frac{1-p}{4b}\!
\left[
\operatorname{erf}\!\!\left(\frac{x+b+2t\Delta}
{2\sqrt{2\gamma_p t}}\right)
-
\operatorname{erf}\!\!\left(\frac{x-b+2t\Delta}
{2\sqrt{2\gamma_p t}}\right)
\right].
\]
Again $\operatorname{erf}(x)$ is the usual error function. The probability density and the population imbalance are given by 
\begin{equation}\label{PDF3}
P(t,x) = \operatorname{Tr}[\rho(t,x)]=\rho_{11}(t,x)+\rho_{22}(t,x),\quad Q(t,x) = \rho_{11}(t,x)-\rho_{22}(t,x).
\end{equation}
Figure \ref{fig3} shows the probability distribution and the population imbalance of the open quantum Brownian particle for different values of the parameters and time evolution.
\begin{figure}[!ht]
    \centering
    \begin{subfigure}[b]{0.5\textwidth}
        \centering
        \includegraphics[height=3.50in]{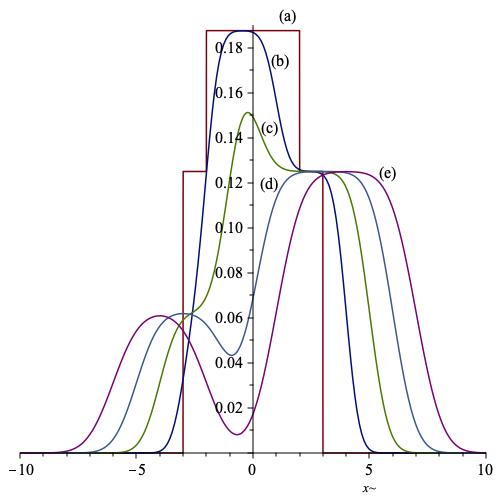}
    \end{subfigure}%
    ~ 
    \begin{subfigure}[b]{0.5\textwidth}
        \centering
        \includegraphics[height=3.50in]{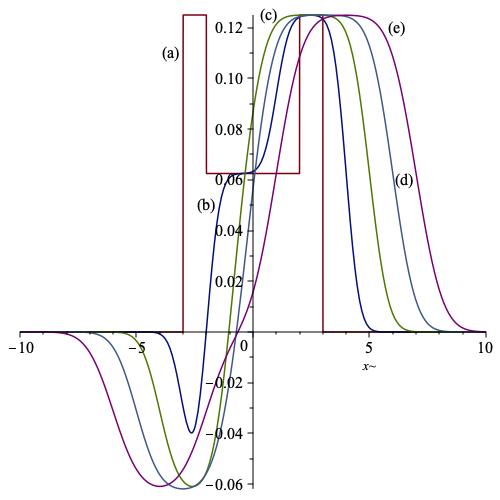}
           \end{subfigure}
  \caption{The probability distribution \eqref{PDF3} (left) and the population imbalance (right) of the OQBM for different moments of time. The initial distribution is given by \eqref{iccsGau3}. The curves from (a) to (e) corresponds to times $0, 50, 100, 150, 200$, respectively. The other parameters are chosen to be $p=3/4,a=3,b=2,\gamma_p=10^{-3},\Delta=10^{-2}, \gamma_z=10^{-3}$.}
\label{fig3}
\end{figure}


\end{enumerate}

It is observed in all cases that, irrespective of the initial conditions, the probability density approaches a superposition of two Gaussian distributions as $t\to\infty$.

\section{Explicit solutions for $\Delta=0$}\label{sec5}

In the specific case where $\Delta=0$, the master equation \eqref{eq:13} simplifies significantly. The term responsible for the ``open'' dynamics vanishes, decoupling the internal and external degrees of freedom of the Brownian particle. Even in this simplified scenario, the generalized parabolic system \eqref{GPS} provides a correct description, allowing for the explicit computation of its matrix-valued Green's function.

In this situation the expression of the eigenvalues and eigenvectors simplify considerably. Indeed, the eigenvalues are
\begin{align*}
\lambda_1(\xi) &= -2\gamma_p \, \xi^2, \\
\lambda_2(\xi) &= -2\gamma_p \xi^2 - \gamma_z +\sqrt{\gamma_z^2 - 4\Omega^2}, \\
\lambda_3(\xi) &= -2\gamma_p \xi^2 - \gamma_z - \sqrt{\gamma_z^2 - 4\Omega^2},
\end{align*}
while the corresponding diagonalization matrix $U(\xi) = U$ is constant:
    \[
    U = \begin{pmatrix}
        1 & 0 & 0 \\
        0 & \dfrac{1}{4\Omega} \left( \gamma_z - \sqrt{\gamma_z^2 - 4\Omega^2} \right) & \dfrac{1}{4\Omega} \left( \gamma_z + \sqrt{\gamma_z^2 - 4\Omega^2} \right) \\
        0 & 1 & 1
    \end{pmatrix}.
    \]
Depending on the relation between $\gamma_z$ and $\Omega$ we will have different matrix-valued Green's functions. From now on we will use the following notation
$$
\omega_{\pm}=\sqrt{\pm\gamma_z^2\mp4\Omega^2}.
$$

\begin{itemize}
\item \textbf{Case $\gamma_z>2\Omega>0$}. We can easily compute $e^{tQ(\xi)}$, which is given by
\[
e^{tQ(\xi)}=e^{-2\gamma_p \xi^2 t} \begin{pmatrix}
1& 0 & 0 \\
0 & e^{- \gamma_z t} \left( \cosh\left( t \omega_+ \right) - \dfrac{\gamma_z}{\omega_+} \sinh\left( t \omega_+\right) \right) & \dfrac{\Omega}{\omega_+}e^{- \gamma_z t}\sinh\left( t \omega_+ \right) \\
0 & -\dfrac{4\Omega}{\omega_+} e^{- \gamma_z t}\sinh\left( t \omega_+ \right) & e^{- \gamma_z t} \left( \cosh\left( t \omega_+ \right) + \dfrac{\gamma_z}{\omega_+} \sinh\left( t \omega_+ \right) \right)
\end{pmatrix}.
\]
By performing the inverse Fourier transform to $e^{t Q(\xi)}$ we obtain
matrix-valued Green's function \eqref{genGreen}, given by
\[
G(t,x) = \frac{e^{-\frac{x^2}{8t\gamma_p}}}{2\sqrt{2\pi\gamma_pt}}
\begin{pmatrix}
1 & 0 & 0 \\[1em]
0 & 
\displaystyle e^{-\gamma_zt}\left(\cosh(\omega_+ t)-\frac{\gamma_z}{\omega_+}\sinh(\omega_+ t)\right)
& 
\displaystyle \frac{\Omega}{\omega_+} e^{-\gamma_zt} \sinh(\omega_+ t) \\[1em]
0 & 
\displaystyle -\frac{4\Omega}{\omega_+} e^{-\gamma_zt} \sinh(\omega_+ t) 
& 
\displaystyle e^{-\gamma_zt} \left(\cosh(\omega_+ t)+\frac{\gamma_z}{\omega_+}\sinh(\omega_+ t)\right)
\end{pmatrix}.
\]
\item \textbf{Case $0<\gamma_z<2\Omega$}. We can easily compute $e^{tQ(\xi)}$, which is given in this case by
\[
e^{tQ(\xi)}=e^{-2\gamma_p \xi^2 t}\begin{pmatrix}
1 & 0 & 0 \\
0 & e^{- \gamma_z t} \left( \cos(\omega_- t) - \dfrac{\gamma_z}{\omega_-} \sin(\omega_- t) \right) & \dfrac{\Omega}{\omega_-}e^{- \gamma_z t} \sin(\omega_- t) \\
0 & -\dfrac{4\Omega}{\omega_-} e^{- \gamma_z t}\sin(\omega_- t) & e^{- \gamma_z t} \left( \cos(\omega_- t) + \dfrac{\gamma_z}{\omega_-} \sin(\omega_- t) \right)
\end{pmatrix},
\]
By performing the inverse Fourier transform to $e^{t Q(\xi)}$ we obtain
matrix-valued Green's function \eqref{genGreen}, given by
\begin{equation}\label{GGd0}
G(t,x)=\frac{e^{-\frac{x^2}{8t\gamma_p}}}{2\sqrt{2\pi \gamma_pt}}\begin{pmatrix}
1& 0 & 0 \\[1em]
0 & \displaystyle e^{-\gamma_zt}\left(\cos(\omega_- t)-\dfrac{\gamma_z}{\omega_-}\sin(\omega_- t)\right)
&\displaystyle \frac{\Omega}{\omega_-} e^{-\gamma_zt}\sin(\omega_- t) \\[1em]
0 & -\displaystyle\frac{4\Omega}{\omega_-} e^{-\gamma_zt}\sin(\omega_- t) & \displaystyle e^{-\gamma_zt}\left(\cos(\omega_- t)+\dfrac{\gamma_z}{\omega_-}\sin(\omega_- t)\right)
\end{pmatrix}.
\end{equation}

\item \textbf{Case $\gamma_z=2\Omega>0$}. In this specific limit, the matrix exponential $e^{tQ(\xi)}$ can be computed directly. The resulting expression is:
$$
e^{tQ(\xi)}=
e^{-2\gamma_p \xi^2 t}
\begin{pmatrix}
1 & 0 & 0 \\[4pt]
0 & \left( 1-2\Omega t  \right)e^{-2\Omega t}  & \Omega t e^{-2\Omega t} \,\\[4pt]
0 & -4\Omega t\, e^{-2\Omega t} &  \left( 1 + 2\Omega t \right)e^{-2\Omega t}
\end{pmatrix}.
$$
By performing the inverse Fourier transform to $e^{t Q(\xi)}$ we obtain
matrix-valued Green's function \eqref{genGreen}, given by
$$
G(t,x)=\frac{e^{-\frac{x^2}{8t\gamma_p}}}{2\sqrt{2\pi\gamma_pt}}
\begin{pmatrix}
1& 0 & 0 \\[8pt]
0 & (1-2t\Omega)e^{-2t\Omega}&
t\Omega e^{-2t\Omega} \\[8pt]
0 & -4t\Omega e^{-2t\Omega} &
(1+2t\Omega)e^{-2t\Omega}
\end{pmatrix}.
$$
\end{itemize}

From the previous analysis, we can clearly see that the Green's function is given by a scalar Gaussian factor multiplied by a matrix that depends only on $t$ (and the free parameters). Therefore, all computations required to obtain the explicit solutions are simpler than in the case $\Omega = 0$.

\medskip

The structure of the matrix-valued Green's function implies that the probability density of the process, $P(t,x)$, is given by the convolution of the initial probability density with a Gaussian kernel:
$$
 P(t,x) = \frac{1}{2\sqrt{2 \pi \gamma_p t}} \int_{-\infty}^\infty e^{-\frac{(x-y)^2}{8t\gamma_p}} \left( \psi_{11}(y)+\psi_{22}(y) \right) dy,
$$
where $\psi_{11}(x)+\psi_{22}(x)$ is the initial probability density from \eqref{iccs}. By the properties of convolution with a Gaussian (or heat) kernel, it follows that for sufficiently large $t$, the distribution $P(t,x)$ will asymptotically approach a Gaussian profile, so there is no need to plot $P(t,x)$ in this situation. However, the population imbalance $Q(t,x)=\rho_{11}(t,x)-\rho_{22}(t,x)$ can offer some insights, since it is linked to the other entries of the matrix-valued Green's function, particularly at certain times. For simplicity we will focus on the case $\gamma_z^2<4\Omega^2$, in which case we have to use \eqref{GGd0}. Then we have
\begin{equation}\label{PIqq}
\begin{split}
 Q(t,x) &= -\frac{2\Omega e^{-\gamma_zt}\sin(\omega_-t)}{\omega_-\sqrt{2 \pi \gamma_p t}} \int_{-\infty}^\infty e^{-\frac{(x-y)^2}{8t\gamma_p}}\Im(\psi_{12}(y))dy\\
 &\qquad \frac{e^{-\gamma_zt}}{2\sqrt{2\pi\gamma_pt}}\left(\cos(\omega_-t)+\frac{\gamma_z}{\omega_-}\sin(\omega_-t)\right)\int_{-\infty}^\infty e^{-\frac{(x-y)^2}{8t\gamma_p}}\left(\psi_{11}(y)-\psi_{22}(y)\right)dy.
 \end{split}
\end{equation}
For instance, if we take a Gaussian initial condition like in \eqref{iccsGau} we have that $\Im(\psi_{12})=0$ and then

\begin{equation}\label{PIqqq}
 Q(t,x)  = \frac{e^{-\gamma_z t}\left(\gamma_z \sin(\omega_- t) + \omega_- \cos(\omega_- t)\right)}{\sqrt{2\pi}\,\omega_-}
\left(
\frac{p \,\exp\!\left(- \tfrac{x^2}{2(4\gamma_p t + \sigma_1^2)}\right)}{\sqrt{4\gamma_p t + \sigma_1^2}}
-\frac{(1-p) \,\exp\!\left( - \tfrac{x^2}{2(4\gamma_p t + \sigma_2^2)}\right)}{\sqrt{4\gamma_p t + \sigma_2^2}}
\right).
\end{equation}
In Figure \ref{fig6} (left) we have plotted this $Q(t,x)$ for different values of the parameters and time evolution. Note that the values of the plots always oscillate between being strictly positive or strictly negative, and that for certain values of time, these plots vanish. It is easy to see that the times at which this happens are given by the zeros $\tau_n$ of the equation $\gamma_z \sin(\omega_- t) + \omega_- \cos(\omega_- t)=0$, which are given by
$$
\tau_n=n\pi-\dfrac{1}{\omega_-}\arctan\left(\dfrac{\omega_-}{\gamma_z}\right),\quad n\in\mathbb{N}.
$$
For the values of the parameters in Figure \ref{fig6}, the first of these zeros is located at 
$$
\tau_1=\frac{1000 \,\sqrt{399}}{399}\left(\dfrac{\pi}{2} + \arctan\!\Big(\dfrac{1}{\sqrt{399}}\Big)\right)\approx81.1423506200.
$$
This situation no longer arises if we require that $\Im(\psi_{12}) \neq 0$. For example, consider the initial condition
\begin{equation}\label{iccsGau22}
\rho(0,x)=\frac{1}{\sqrt{2\pi}\,\sigma} \, e^{-x^2/(2\sigma^2)}\begin{pmatrix}p&\mu\sqrt{p(1-p)}e^{ikx}\\
\mu\sqrt{p(1-p)}e^{-ikx}&1-p
\end{pmatrix},\; k\in\mathbb{R},\; 0<p,\mu<1,\;\sigma>0.
\end{equation}
Then we have that $\Im(\psi_{12}) \neq 0$ and the first term in \eqref{PIqq} no longer vanishes. However, we can obtain an explicit expression of $Q(t,x)$, given in this case by
\begin{equation}\label{PIqqqq}
\begin{split}
Q(t,x)&=
-\frac{2\Omega\,\mu\,\sqrt{p(1-p)}\,\sin(\omega_- t)}{\omega_- \sqrt{2\pi}\,\sqrt{4\gamma_p t + \sigma^2}} \,
\exp\!\Big(-\frac{16 \gamma_p k^2 \sigma^2 t + x^2}{2(8 \gamma_p t + \sigma^2)}\Big)\,
\sin\Big(\frac{k x \sigma^2}{4 \gamma_p t + \sigma^2}\Big)\\
&\quad +
\frac{(2p-1)}{\omega_- \sqrt{2\pi}\,\sqrt{4\gamma_p t + \sigma^2}} \left(\gamma_z \sin(\omega_- t) + \omega_- \cos(\omega_- t)\right)\,
\exp\!\Big(-\frac{16 \gamma_p \gamma_z t^2 + 2 \gamma_z \sigma^2 t + x^2}{2(4 \gamma_p t + \sigma^2)}\Big).
\end{split}
\end{equation}
In Figure \ref{fig6} (right) we have plotted this $Q(t,x)$ for different values of the parameters and time evolution. Note that the values of the plots are different from the values on the left.

\begin{figure}[!ht]
    \centering
    \begin{subfigure}[b]{0.5\textwidth}
        \centering
        \includegraphics[height=3.50in]{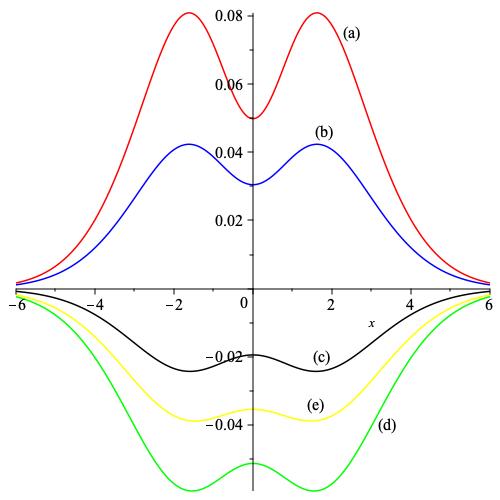}
    \end{subfigure}%
    ~ 
    \begin{subfigure}[b]{0.5\textwidth}
        \centering
        \includegraphics[height=3.50in]{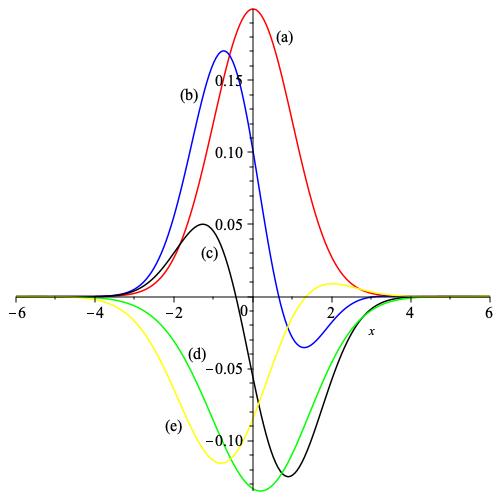}
           \end{subfigure}
  \caption{The population imbalance (\eqref{PIqqq} on the left and \eqref{PIqqqq} on the right) of the OQBM for different moments of time. On the left, the initial distribution is given by \eqref{iccsGau}. On the right, by \eqref{iccsGau22}. The curves from (a) to (e) corresponds to times $0, 50, 100, 150, 200$, respectively. The other parameters are chosen to be $p=3/4,\sigma_1=2, \sigma_2=1,\gamma_p=10^{-3},\gamma_z=10^{-3},\Omega=10^{-2}$ (on the left) and $\mu=4/5, k=1, p=3/4,\sigma=1,\gamma_p=10^{-3},\gamma_z=10^{-3},\Omega=10^{-2}$ (on the right).}
\label{fig6}
\end{figure}

\section{Explicit solutions for $\gamma_z=0$}\label{sec6}

Setting $\gamma_z=0$ in the master equation \eqref{eq:13} effectively eliminates the coherent component of the dissipative dynamics associated with the internal degree of freedom of the Brownian particle, while retaining the unitary contribution. In this case, the matrices $B$ and $C$ in \eqref{BBCC} no longer commute.

\medskip

The expression of the eigenvalues and eigenvectors simplify again considerably, but now the matrix $U(\xi)$ will have a dependence on $\xi$. Let us use the following notation
$$
\omega(\xi)=2\sqrt{\Delta^2\xi^2 + \Omega^2}.
$$
Then the eigenvalues of $Q(\xi)$ can be written as:
 \begin{align*}
     \lambda_3(\xi) &= -2\gamma_p \xi^2,\\
      \lambda_1(\xi) &= -2\gamma_p \xi^2 + i \omega(\xi), \\
        \lambda_2(z) &= -2\gamma_p \xi^2 - i \omega(\xi),
    \end{align*}   
    and the matrix $U(\xi)$ now depends on $\xi$:
    \[
    U(\xi) = \begin{pmatrix}
      \displaystyle \frac{2i\Omega}{\Delta \xi}&  \displaystyle -\frac{2\Delta \xi}{\omega(\xi)} & \displaystyle \frac{2\Delta \xi}{\omega(\xi)}  \\[1em]
       1& \displaystyle -\frac{i\Omega}{\omega(\xi)} & \displaystyle \frac{i\Omega}{\omega(\xi)} \\[1em]
     0&   1 & 1 
    \end{pmatrix}.
    \]
    
With this information we may compute $e^{t Q(\xi)}$, which is given by
\begin{equation}\label{etQz}
e^{tQ(\xi)} = e^{-2 t \xi^2 \gamma_p}
\begin{pmatrix}
\displaystyle \frac{4\Omega^2 +  4\xi^2 \Delta^2 \cos(t\omega(\xi))}{\omega(\xi)^2} &
\displaystyle -\frac{8i \Delta \xi \Omega ( \cos(t\omega(\xi)) - 1 )}{\omega(\xi)^2} &
\displaystyle - \frac{2i\xi \Delta \sin(t\omega(\xi))}{\omega(\xi)} \\[0.7em]

\displaystyle \frac{2i\Delta \xi \Omega ( \cos(t\omega(\xi)) - 1 )}{\omega(\xi)^2} &
\displaystyle \frac{4\xi^2 \Delta^2 + 4\Omega^2 \cos(t\omega(\xi))}{\omega(\xi)^2} &
\displaystyle \frac{\Omega \sin(t\omega(\xi))}{\omega(\xi)} \\[0.7em]

\displaystyle -\frac{2i\xi \Delta \sin(t\omega(\xi))}{\omega(\xi)} &
\displaystyle -\frac{4\Omega \sin(t\omega(\xi))}{\omega(\xi)} &
\displaystyle \cos(t\omega(\xi))
\end{pmatrix}.
\end{equation}
To obtain the matrix-valued Green’s function \eqref{genGreen}, we must compute the inverse Fourier transform of each entry of the matrix-valued function above. Unlike the two previous cases, this task is not straightforward. Nevertheless, we succeeded in deriving the matrix-valued Green’s function by employing convolution techniques, as established in the following result.

\begin{pro}
Define the following functions:
\begin{equation}\label{functtt}
\begin{split}
g(t,x)&=\frac{1}{2\sqrt{2\gamma_p\pi t}}e^{-\frac{x^2}{8\gamma_p t}},\\
h_{\pm}(t,x)&=\frac{\Omega}{4\Delta} \, \exp\!\left(\frac{2\Omega^{2}\gamma_p t}{\Delta^{2}}\right)
\left[
   e^{-\frac{\Omega x}{\Delta}} \, \operatorname{erfc}\!\left(\frac{4\gamma_p\Omega t - \Delta x}{2\Delta\sqrt{2\gamma_p t}}\right)
   \pm e^{\tfrac{\Omega x}{\Delta}} \, \operatorname{erfc}\!\left(\frac{4\gamma_p\Omega t + \Delta x}{2\Delta\sqrt{2\gamma_p t}}\right)
\right],\\
\kappa_0(t,x)&=\frac{1}{4\Delta}J_0\left(\frac{\Omega}{\Delta}\sqrt{4\Delta^2t^2-x^2}\right)\chi_{|x|<2\Delta t},\\
\kappa_1(t,x)&=\frac{1}{2}\left[\delta(x-2\Delta t)+\delta(x+2\Delta t)\right]-\frac{t\Omega}{\sqrt{4\Delta^2t^2-x^2}}J_1\left(\frac{\Omega}{\Delta}\sqrt{4\Delta^2t^2-x^2}\right)\chi_{|x|<2\Delta t},
\end{split}
\end{equation}
where $\operatorname{erfc}(x)$ denotes the complementary error function, $J_\alpha(z)$ the Bessel function of the first kind, $\delta(x)$ is the Dirac delta and $\chi_A$ the indicator function. Then, the matrix-valued Green's function \eqref{genGreen} can be written as
\begin{equation}\label{Greez}
G(t,x)=
\begin{pmatrix}
h_++(g-h_+)\ast\kappa_1& -2h_-+2(h_-\ast\kappa_1) & \frac{\Delta}{2\gamma_pt}(xg\ast\kappa_0) \\[8pt]
\frac{1}{2}h_--\frac{1}{2}(h_-\ast\kappa_1)  &g-h_++h_+\ast\kappa_1&\Omega(g\ast\kappa_0) \\[8pt]
\frac{\Delta}{2\gamma_pt}(xg\ast\kappa_0) & -4\Omega(g\ast\kappa_0) & g\ast\kappa_1
\end{pmatrix}.
\end{equation}
where $\ast$ denotes the usual convolution operator.
\end{pro}
\begin{proof}
Let $\mathcal{F}$ and $\mathcal{F}^{-1}$ denote the usual Fourier and inverse Fourier transforms, i.e.
$$
\mathcal{F}[f](\xi)=\widehat{f}(\xi)=\int_{\mathbb{R}}f(x)e^{-ix\xi}dx,\quad \mathcal{F}^{-1}[h](x)=\frac{1}{2\pi}\int_{\mathbb{R}}h(\xi)e^{ix\xi}d\xi.
$$
Then we have that the functions in \eqref{functtt} can be written as
$
g(t,x)=\mathcal{F}^{-1}\left[e^{-2\gamma_p t\xi^2}\right](x),
$
$$
h_+(t,x)=\mathcal{F}^{-1}\left[\frac{4\Omega^2}{\omega(\xi)^2}e^{-2\gamma_p t\xi^2}\right](x),\quad  h_-(t,x)=\mathcal{F}^{-1}\left[\frac{4\Omega\Delta\xi}{i\omega(\xi)^2}e^{-2\gamma_p t\xi^2}\right](x),
$$
and
\begin{equation*}\label{kkff}
\kappa_0(t,x)=\mathcal{F}^{-1}\left[\frac{\sin(t\omega(\xi))}{\omega(\xi)}\right](x),\quad \kappa_1(t,x)=\mathcal{F}^{-1}\left[\cos(t\omega(\xi))\right](x).
\end{equation*}
Additionally, we also have
$$
\mathcal{F}^{-1}\left[\frac{\Delta\xi}{i}e^{-2\gamma_p t\xi^2}\right](x)=\frac{\Delta x}{4\gamma_pt}g(t,x),\quad \mathcal{F}^{-1}\left[\frac{4\Delta^2\xi^2}{\omega(\xi)^2}e^{-2\gamma_p t\xi^2}\right](x)=g(t,x)-h_+(t,x).
$$
Applying the convolution theorem, $\mathcal{F}\left[f\ast g\right](\xi)=\widehat{f}(\xi)\widehat{g}(\xi)$ and considering the previous relations to the expression for the matrix exponential in \eqref{etQz} immediately yields the form of the Green's function in \eqref{Greez}.
\end{proof}

\begin{remark}
Let us give a brief justification of the computation of the Fourier inverse transforms of $h_{\pm}(t,x)$ and $\kappa_i(t,x), i=0,1,$ in \eqref{functtt}. For $h_{\pm}(t,x)$, using the standard formula
\(\mathcal{F}^{-1}\!\big[(\xi^2+a^2)^{-1}\big](x)=(2a)^{-1}e^{-a|x|}\) we obtain
$$
\mathcal{F}^{-1}\!\left[\frac{4\Omega^2}{\omega(\xi)^2}\right](x)=\mathcal{F}^{-1}\!\left[\frac{\Omega^2/\Delta^2}{\xi^2+(\Omega/\Delta)^2}\right](x)=\frac{\Omega}{2\Delta}\,e^{-(\Omega/\Delta)|x|}.
$$
By the convolution theorem, the inverse transform is the convolution of the two inverse transforms. Using that $\mathcal{F}^{-1}\!\left[e^{-2\gamma_p t\,\xi^2}\right](x)=g(t,x)$ we have
\[
\mathcal{F}^{-1}\!\left[\frac{4\Omega^2}{\omega(\xi)^2}\,e^{-2\gamma_p t\,\xi^2}\right](x)
=\left(\frac{1}{\sqrt{8\pi\gamma_p t}}e^{-\frac{(\,\cdot\,)^2}{8\gamma_p t}}\right)
* \left(\frac{\Omega}{2\Delta}e^{-(\Omega/\Delta)|\cdot|}\right)(x)
=\frac{\Omega}{2\Delta}\int_{\mathbb{R}}
\frac{e^{-\frac{(x-y)^2}{8\gamma_p t}}}{\sqrt{8\pi \gamma_p t}}\,e^{-(\Omega/\Delta)|y|}\,dy.
\]
Splitting the integral at \(y=0\), completing the square and using the formula
\[
\int_{0}^{\infty}e^{-(\Omega/\Delta)y}\exp\!\left(-\frac{(y-x)^2}{8\gamma_p t}\right)dy
=\sqrt{2\pi\gamma_p t}\;e^{\frac{2\Omega^2\gamma_pt}{\Delta^2}}e^{-(\Omega/\Delta)x}\operatorname{erfc}\!\left(\frac{4\gamma_p t\Omega-\Delta x}{2\Delta\sqrt{2\gamma_p t}}\right),
\]
we get the result. Similarly for $h_-$, where now we use
$$
\mathcal{F}^{-1}\!\left[\frac{4\Omega\Delta\xi}{i\omega(\xi)^2}\right](x)=\frac{\Omega}{2\Delta}\,\sgn(x)\,e^{-(\Omega/\Delta)|x|}.
$$
On the other hand, the function $\kappa_0(t,x)$ is in fact the solution of the 1D Klein-Gordon equation with initial conditions:
\[
\partial_t^2 \kappa_0(t,x)-4\Delta^2\partial_x^2 \kappa_0(t,x)+4\Omega^2 \kappa_0(t,x)=0,
\quad \kappa_0(0,x)=0,\quad \partial_t \kappa_0(0,x)=\delta(x).
\]
Indeed, applying Fourier transform to the previous PDE gives the ODE
\[
\partial_t^2\widehat\kappa_0(t,\xi) + \omega(\xi)^2\,\widehat \kappa_0(t,\xi)=0,\quad \widehat\kappa_0(0,\xi)=0,\quad \partial_t\widehat\kappa_0(0,\xi)=1,
\]
since $\mathcal F[\delta]=1$. The general solution is given by 
\[
\widehat\kappa_0(t,\xi)=A(\xi)\cos(t\omega(\xi))+B(\xi)\sin(t\omega(\xi)),
\]
and the initial conditions give $A(\xi)=0,\; B(\xi)=\dfrac{1}{\omega(\xi)}$. Thus
\[
\displaystyle \widehat\kappa_0(t,\xi)=\frac{\sin\!\big(t\omega(\xi)\big)}{\omega(\xi)}.
\]
Inverting the Fourier transform gives
\begin{equation}\label{k0fou}
\kappa_0(t,x)=\frac{1}{2\pi}\int_{-\infty}^{\infty}\frac{\sin\!\big(\omega(\xi)t\big)}{\omega(\xi)}\,e^{i x\xi}\,d\xi.
\end{equation}
Since the integrand is an even function, and using the classical transform identity (which can be found in \cite[formula 3.876.1]{GrRy} or \cite[p.26, formula (30)]{Erdely})
\begin{equation*}\label{J0form}
\displaystyle
\int_{0}^{\infty} \frac{\sin\!\big(a\sqrt{\xi^2+b^2}\big)}{\sqrt{\xi^2+b^2}}\cos(x\xi)\,d\xi
= \dfrac{\pi}{2}\,J_0\!\big(b\sqrt{a^2-x^2}\big)\chi_{|x|<a},
\end{equation*}
we get $\kappa_0(t,x)$ as written in \eqref{functtt}.

For $\kappa_1(t,x)$ we also have that it is the solution of the 1D Klein-Gordon equation with different initial conditions:
\[
\partial_t^2 \kappa_1(t,x)-4\Delta^2\partial_x^2 \kappa_1(t,x)+4\Omega^2 \kappa_1(t,x)=0,
\quad \kappa_1(0,x)=\delta(x),\quad \partial_t \kappa_1(0,x)=0.
\]
Applying Fourier transform now we get
\[
\displaystyle \widehat\kappa_1(t,\xi)=\cos\!\big(t\omega(\xi)\big)\; \Longrightarrow\; \kappa_1(t,x)=\frac{1}{2\pi}\int_{-\infty}^{\infty} e^{i x\xi}\cos\!\big(t\omega(\xi)\big)d\xi=\frac{1}{\pi}\int_{0}^{\infty} \cos\!\big(t\omega(\xi)\big)\cos(i x\xi)d\xi.
\]
We observe from \eqref{k0fou} that, at least formally, 
\[
\kappa_1(t,x) = \partial_t \kappa_0(t,x).
\]
Differentiating $\kappa_0(t,x)$ in \eqref{functtt} with respect to $t$ 
in the sense of distributions requires care. 
One must distinguish between the derivative of the smooth factor valid in the interior region $|x|<2\Delta t$, and the contribution arising from the derivative of the indicator function at the boundary $|x|=2\Delta t$, which can be expressed using the Heaviside function as $H(2\Delta t - |x|)$.

Using \(J_0'(z)=-J_1(z)\) and
\[
\frac{d}{dt}\left[\sqrt{t^{2}-\frac{x^{2}}{4\Delta^2}}\right]=\frac{t}{\sqrt{t^{2}-\frac{x^{2}}{4\Delta^2}}}\quad (|x|<2\Delta t),\quad
\frac{d}{dt}H(2\Delta t-|x|)=2\Delta\,\delta(2\Delta t-|x|),
\]
we obtain (distributionally)
$$
\partial_t \kappa_0(t,x)= \frac{1}{4\Delta}\left[ -J_1\left(2\Omega \sqrt{t^{2}-\frac{x^{2}}{4\Delta^2}}\right)\cdot\frac{2\Omega t}{\sqrt{t^{2}-\frac{x^{2}}{4\Delta^2}}}\;H(2\Delta t-|x|) \;+\; J_0(0)\,2\Delta\,\delta(2\Delta t-|x|)\right].
$$
Since $J(0)=1$, we obtain $\kappa_1(t,x)$ in \eqref{functtt}.

\end{remark}

\begin{remark}
We have expressed the matrix-valued Green's function \eqref{Greez} 
as a combination and convolution of the functions in \eqref{functtt}. 
Each entry admits an integral representation. 
For example, the $(2,3)$ entry can be written as
\[
\Omega \bigl(g \ast \kappa_0\bigr)(t,x)
= \frac{\Omega}{4\Delta\sqrt{2\pi\gamma_p t}}
\int_{-2\Delta t}^{2\Delta t}
\exp\!\left(-\frac{(x-y)^2}{8\gamma_p t}\right)
J_0\!\left(\frac{\Omega}{\Delta}
\sqrt{4\Delta^2t^2-y^2}\right)\,dy.
\]
Alternatively, the same entry admits a representation in terms of a cosine transform, since all functions involved are even:
\[
\frac{1}{\pi}
\int_{0}^{\infty}
\exp\!\bigl(-2\gamma_p t \xi^2\bigr)\,
\frac{\Omega \sin\bigl(t\omega(\xi)\bigr)}{\omega(\xi)}
\cos(x\xi)\,d\xi.
\]
These two integral representations are connected by the convolution theorem. In principle, convolutions involving the functions $\kappa_0$ and $\kappa_1$ offer an advantage, as they lead to integrals over bounded intervals, whereas the cosine transform requires evaluating an improper integral. Depending on the initial conditions of the problem, it may be preferable to use either the convolution form or the cosine transform form.
\end{remark}

Let us study now one specific example with the following initial conditions:
\begin{equation}\label{iccslapz}
\rho(0,x)=f_L(x)\begin{pmatrix}p&\sqrt{p(1-p)}(r+iq)\\
\sqrt{p(1-p)}(r-iq)&1-p
\end{pmatrix},\; 0<p<1,\; r,q\in\mathbb{R},\; r^2+q^2\leq1,
\end{equation}
where 
$$
f_L(x)=\frac{\Omega}{2\Delta}e^{-\frac{\Omega}{\Delta}|x|},\quad x\in\mathbb{R},
$$
is the Laplace distribution with scale parameter $\Delta/\Omega$. This choice considerably simplifies computations. Although the example can also be computed using a different free scale parameter $a$, the resulting computations are essentially analogous. Since we are dealing with the full matrix-valued function $G(t,x)$ in \eqref{Greez}, we may consider an initial density matrix like \eqref{iccslapz} whose components are not restricted to the main diagonal.

With these initial conditions we have $\vec{u}_0(x)=\frac{\Omega}{2\Delta}e^{-\frac{\Omega}{\Delta}|x|}(1,q\sqrt{p(1-p)},2p-1)^T$. Let $\vec{u}(t,x)=(u_1(t,x),u_2(t,x),u_3(t,x))^T$ be the solution of \eqref{GPS}. Using the Green's function \eqref{Greez}, we obtain
\smallskip
\begin{equation}\label{uz1}
\begin{split}
u_1(t,x)&=h_+\ast f_L+g \ast \kappa_1 \ast f_L-h_+ \ast \kappa_1 \ast f_L\\
&\quad-2q\sqrt{p(1-p)}\left(h_-\ast f_L-h_-\ast\kappa_1\ast f_L\right)+(2p-1)\frac{\Delta}{2\gamma_pt}(xg\ast\kappa_0\ast f_L),\\
u_2(t,x)&=\frac{1}{2}h_-\ast f_L-\frac{1}{2}h_-\ast\kappa_1\ast f_L\\
&\quad+q\sqrt{p(1-p)}\left(g\ast f_L-h_+\ast f_L+h_+\ast\kappa_1\ast f_L\right)
+(2p-1)\Omega\left(g\ast\kappa_0\ast f_L\right),\\
u_3(t,x)&=\frac{\Delta}{2\gamma_pt}\left(xg\ast\kappa_0\ast f_L\right)-4\Omega q\sqrt{p(1-p)}(g\ast\kappa_1\ast f_L)+(2p-1)\left(g\ast\kappa_1\ast f_L\right),
\end{split}
\end{equation}
where $\ast$ denotes the usual convolution operator. These expressions can be further simplified by exploiting the associativity of the convolution operation, together with the following identities:
$$
f_L\ast f_L=\frac{1}{2\Delta}(\Omega|x|+\Delta) f_L,\quad f_L\ast (\sgn(x)f_L)=\frac{x\Omega}{2\Delta}f_L,
$$
$$
g\ast f_L=h_+,\quad g\ast(\sgn(x)f_L)=h_-,\quad xg\ast f_L=\frac{4\gamma_pt\Omega}{\Delta}h_-,
$$
\begin{align*}
\phi_+:=h_+\ast f_L=&
-\frac{\Omega}{8\,\Delta^{3}}
\exp\!\Bigl(\frac{2\Omega^{2} \gamma_p t}{\Delta^{2}}\Bigr)\,
\!\left[
e^{-\Omega x/\Delta}\,(4\Omega^{2} \gamma_p t-\Delta^{2}-\Omega\Delta x)\,
\operatorname{erfc}\!\left(\frac{4\gamma_p\Omega t - \Delta x}{2\Delta\sqrt{2\gamma_p t}}\right)\right.\\
&
\quad+
\left.e^{\Omega x/\Delta}\,(4\Omega^{2} \gamma_p t-\Delta^{2}+\Omega\Delta x)\,
\operatorname{erfc}\!\left(\frac{4\gamma_p\Omega t + \Delta x}{2\Delta\sqrt{2\gamma_p t}}\right)
\right]+\frac{\Omega^{2}}{\Delta^{2}}\sqrt{\frac{\gamma_pt}{2\pi}}\,e^{-x^{2}/(8 \gamma_p t)}.
\end{align*}

\begin{align*}
\phi_-:=h_-\ast f_L=&
\frac{\Omega^{2}}{8\,\Delta^{3}}\,
\exp\!\Bigl(\frac{2\Omega^{2} \gamma_p t}{\Delta^{2}}\Bigr)\,
\left[
e^{\frac{\Omega x}{\Delta}}\,(4\Omega \gamma_p t+\Delta x)\,
\operatorname{erfc}\!\left(\frac{4\gamma_p\Omega t + \Delta x}{2\Delta\sqrt{2\gamma_p t}}\right)\right.\\
&\qquad-\left.
e^{-\frac{\Omega x}{\Delta}}\,(4\Omega \gamma_p t-\Delta x)\,
\operatorname{erfc}\!\left(\frac{4\gamma_p\Omega t - \Delta x}{2\Delta\sqrt{2\gamma_p t}}\right)
\right].
\end{align*}
With this simplification we may write \eqref{uz1} as
\begin{equation}\label{uz2}
\begin{split}
u_1(t,x)&=\phi_++(h_+-\phi_+)\ast\kappa_1-2q\sqrt{p(1-p)}\,(\phi_--\phi_-\ast\kappa_1)+2\Omega(2p-1)h_-\ast\kappa_0,\\
u_2(t,x)&=\frac{1}{2}\left(\phi_--\phi_-\ast\kappa_1\right)+q\sqrt{p(1-p)}\,(h_+-\phi_++\phi_+\ast\kappa_1)+\Omega(2p-1)h_+\ast\kappa_0,\\
u_3(t,x)&=2\Omega h_-\ast\kappa_0-4\Omega q\sqrt{p(1-p)}\,h_+\ast\kappa_0+(2p-1)h_+\ast\kappa_1.
\end{split}
\end{equation}
\begin{remark}
Observe that in the above simplification we avoided the convolution of $f_L$ with the Bessel-type functions $\kappa_0$ and $\kappa_1$.  
However, if we assume that $x > 2\Delta t$, the convolutions $f_L \ast \kappa_0$ and $f_L \ast \kappa_1$ admit a further simplification. Indeed, using \cite[formula 3.876.4]{GrRy}, we obtain
$$
f_L\ast\kappa_0=\mathcal{F}^{-1}\left[\widehat{f_L}\widehat{\kappa_0}\right]=\frac{\Omega}{2}\mathcal{F}^{-1}\left[\frac{\sin(2t\sqrt{\Delta^2\xi^2+\Omega^2})}{(\Delta^2\xi^2+\Omega^2)^{3/2}}\right]=tf_L,\quad\mbox{for}\;\; x>2t\Delta,
$$
and using \cite[formula 3.876.5]{GrRy}, we obtain
$$
f_L\ast\kappa_1=\mathcal{F}^{-1}\left[\widehat{f_L}\widehat{\kappa_1}\right]=\Omega^2\mathcal{F}^{-1}\left[\frac{\cos(2t\sqrt{\Delta^2\xi^2+\Omega^2})}{\Delta^2\xi^2+\Omega^2}\right]=f_L,\quad\mbox{for}\;\; x>2t\Delta.
$$
Therefore, for $x>2t\Delta$, we obtain an explicit expression of the solutions \eqref{uz2}, given in terms only of the functions $h_{\pm}$ (see \eqref{functtt}):
\begin{equation*}
\begin{split}
u_1(t,x)&=h_++2t\Omega(2p-1)h_-,\\
u_2(t,x)&=\left[q\sqrt{p(1-p)}+t\Omega(2p-1)\right]h_+,\\
u_3(t,x)&=2t\Omega h_-+\left[-4\Omega q\sqrt{p(1-p)}+(2p-1)\right]h_+.
\end{split}
\end{equation*}
The formulas contained in \cite{GrRy} are only valid for $x>2t\Delta$. For $x\leq2t\Delta$ we have not been able to find an explicit expression of the convolution functions $f_L\ast\kappa_0$ and $f_L\ast\kappa_1$.
\end{remark}

In this case, the probability density $P(t,x)$ coincides with the function $u_1(t,x)$ in \eqref{uz2}. Let us now write this expression explicitly in order to obtain a more simplified form. From \eqref{functtt} we obtain
\begin{align*}
P(t,x)&=\phi_+(t,x)+\frac{1}{2}\left(h_+(t,x-2t\Delta)-\phi_+(t,x-2t\Delta)+h_+(t,x+2t\Delta)-\phi_+(t,x+2t\Delta)\right)\\
&\qquad-t\Omega\int_{-2t\Delta}^{2t\Delta}\frac{h_+(t,x-y)-\phi_+(t,x-y)}{\sqrt{4t^2\Delta^2-y^2}}J_1\left(\frac{\Omega}{\Delta}\sqrt{4t^2\Delta^2-y^2}\right)dy\\
&\quad\qquad-2q\sqrt{p(1-p)}\left[\phi_-(t,x)+\frac{1}{2}\left(\phi_-(t,x-2t\Delta)+\phi_-(t,x+2t\Delta)\right)\right.\\
&\qquad\qquad\left.+t\Omega\int_{-2t\Delta}^{2t\Delta}\frac{\phi_-(t,x-y)}{\sqrt{4t^2\Delta^2-y^2}}J_1\left(\frac{\Omega}{\Delta}\sqrt{4t^2\Delta^2-y^2}\right)dy\right]\\
&\quad\qquad\qquad+(2p-1)\frac{\Omega}{2\Delta}\int_{-2t\Delta}^{2t\Delta}h_-(t,x-y)J_0\left(\frac{\Omega}{\Delta}\sqrt{4t^2\Delta^2-y^2}\right)dy.
\end{align*}
After the change of variables $y=2t\Delta\cos(\theta)$ in the integrals we obtain a further simplification:
\begin{equation}\label{pdfF}
\begin{split}
P(t,x)&=\phi_+(t,x)+\frac{1}{2}\left(h_+(t,x-2t\Delta)-\phi_+(t,x-2t\Delta)+h_+(t,x+2t\Delta)-\phi_+(t,x+2t\Delta)\right)\\
&\quad-t\Omega\int_{0}^{\pi}\left[h_+(t,x-2t\Delta\cos(\theta))-\phi_+(t,x-2t\Delta\cos(\theta))\right]J_1\left(2t\Omega\sin(\theta)\right)d\theta\\
&\quad\qquad-2q\sqrt{p(1-p)}\left[\phi_-(t,x)+\frac{1}{2}\left(\phi_-(t,x-2t\Delta)+\phi_-(t,x+2t\Delta)\right)\right.\\
&\qquad\qquad\left.+t\Omega\int_{0}^{\pi}\phi_-(t,x-2t\Delta\cos(\theta))J_1\left(2t\Omega\sin(\theta)\right)d\theta\right]\\
&\qquad\qquad\qquad+(2p-1)t\Omega\int_{0}^{\pi}h_-(t,x-2t\Delta\cos(\theta))\sin(\theta)J_0\left(2t\Omega\sin(\theta)\right)d\theta.
\end{split}
\end{equation}
With \eqref{pdfF} we have enough information to plot the probability density function. Figure \ref{fig4} shows the probability distribution of the open quantum Brownian particle for different values of the parameters and time evolution.
\begin{figure}[!ht]
    \centering
    \begin{subfigure}[b]{0.5\textwidth}
        \centering
        \includegraphics[height=3.50in]{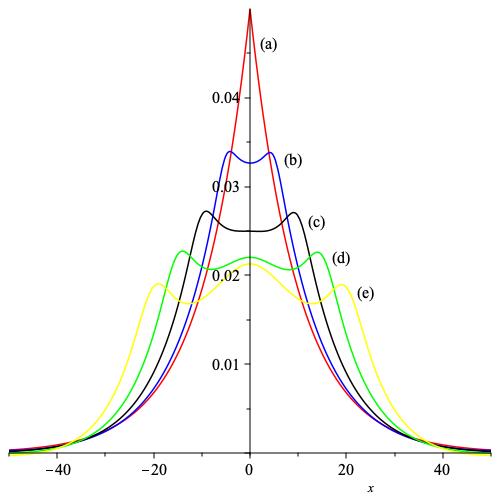}
    \end{subfigure}%
    ~ 
    \begin{subfigure}[b]{0.5\textwidth}
        \centering
        \includegraphics[height=3.50in]{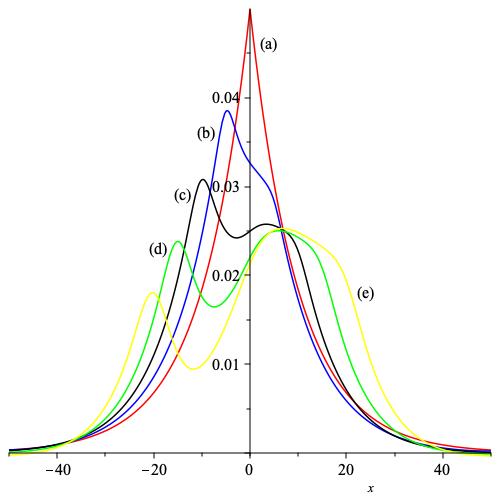}
           \end{subfigure}
  \caption{The probability distribution \eqref{pdfF} of the OQBM for different moments of time. The initial distribution is given by \eqref{iccslapz}. The curves from (a) to (e) corresponds to times $0, 25, 50, 75, 100$, respectively. The graph on the left is for $q=0$, while the graph on the right is for $q=-1/2$. The other parameters are chosen to be $p=1/4,\gamma_p=10^{-2},\Delta=10^{-1}, \Omega=10^{-2}$.}
\label{fig4}
\end{figure}
We observe, unlike the case when $\Omega=0$, that the probability density approaches a superposition of \emph{three} Gaussian distributions as $t\to\infty$.

As for the population imbalance $Q(t,x)=\rho_{11}(t,x)-\rho_{22}(t,x)$, this is given by the function $u_3(t,x)$ in \eqref{uz2}, which in this case it is given, after the same simplifications as before, by:
\begin{equation}\label{popimbz}
\begin{split}
Q(t,x)&=\frac{2p-1}{2}\left(h_+(t,x-2t\Delta)+h_+(t,x+2t\Delta)\right)\\
&\quad+t\Omega\int_{0}^{\pi}h_-(t,x-2t\Delta\cos(\theta))\sin(\theta)J_0\left(2t\Omega\sin(\theta)\right)d\theta\\
&\qquad-2tq\Omega\sqrt{p(1-p)}\int_0^\pi h_+(t,x-2t\Delta\cos(\theta))\sin(\theta)J_0\left(2t\Omega\sin(\theta)\right)d\theta\\
&\quad \qquad-(2p-1)t\Omega\int_{0}^{\pi}h_+(t,x-2t\Delta\cos(\theta))J_1\left(2t\Omega\sin(\theta)\right)d\theta.
\end{split}
\end{equation}
In Figure \ref{fig5} we have plotted the population imbalance of the open quantum Brownian particle for different values of the parameters and time evolution.
\begin{figure}[!ht]
    \centering
    \begin{subfigure}[b]{0.5\textwidth}
        \centering
        \includegraphics[height=3.50in]{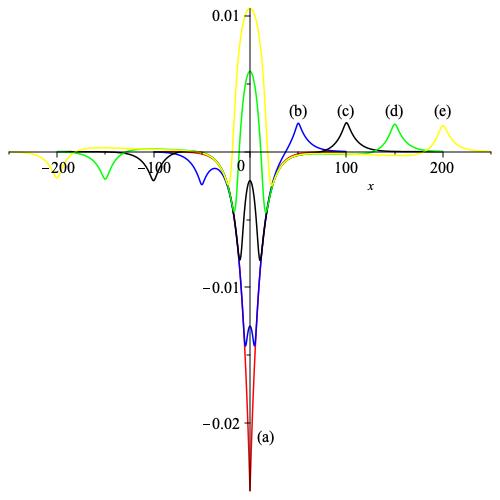}
    \end{subfigure}%
    ~ 
    \begin{subfigure}[b]{0.5\textwidth}
        \centering
        \includegraphics[height=3.50in]{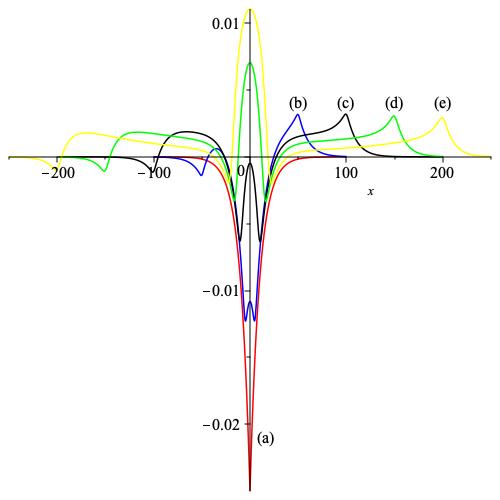}
           \end{subfigure}
  \caption{The population imbalance \eqref{popimbz} of the OQBM for different moments of time. The initial distribution is given by \eqref{iccslapz}. The curves from (a) to (e) corresponds to times $0, 25, 50, 75, 100$, respectively. The graph on the left is for $q=0$, while the graph on the right is for $q=-1/2$. The other parameters are chosen to be $p=1/4,\gamma_p=10^{-2},\Delta=10^{-1}, \Omega=10^{-2}$.}
\label{fig5}
\end{figure}

\section{Appendix: The Born-Markov approximation}

For completeness, we recall the Born-Markov approximation, see e.g. \cite{car}. Begin with a Hamiltonian in the general form
$$H=H_S+H_R+H_{SR},$$
where $H_S$ and $H_R$ are Hamiltonians for system $S$ and reservoir $R$, respectively, and $H_{SR}$ is an interaction Hamiltonian. We would like to find information on system $S$ without requiring detailed information on the composite system $S\otimes R$. If $\chi(t)$ is the density operator for $S\otimes R$, define the reduced density 
$$\rho(t)=\mathrm{Tr}_R(\chi(t))$$
We would like to obtain an equation for $\rho(t)$ with the properties of $R$ entering only as parameters. In order to do this we proceed as follows: starting from Schr\"odinger's equation,
$$\dot{\chi}(t)=\frac{1}{i\hbar}[H,\chi],$$
write
$$\tilde{\chi}(t)=\exp{\Big[\frac{i}{\hbar}(H_S+H_R)t\Big]}\chi(t)\exp{\Big[-\frac{i}{\hbar}(H_S+H_R)t\Big]},$$
$$\tilde{H}_{SR}(t)=\exp{\Big[\frac{i}{\hbar}(H_S+H_R)t\Big]}H_{SR}\exp{\Big[-\frac{i}{\hbar}(H_S+H_R)t\Big]},$$
so that
$$\dot{\tilde{\chi}}(t)=\frac{1}{i\hbar}[\tilde{H}_{SR},\tilde{\chi}].$$
Integrating and substituting for $\tilde{\chi}(t)$ in the commutator gives
\begin{equation*}
\dot{\tilde{\chi}}=\frac{1}{i\hbar}[\tilde{H}_{SR}(t),\chi(0)]-\frac{1}{\hbar^2}\int_0^t [\tilde{H}_{SR}(t),[\tilde{H}_{SR}(t'),\tilde{\chi}(t')]]\;dt',
\end{equation*}
which is Schr\"odinger's equation in integro-differential form. This form allows us to identify reasonable approximations.

\medskip

Now we assume that the interaction begins at time $t=0$ and that no correlations exist between $S$ and $R$ at this time. Therefore, we can write $\chi(0)=\rho(0)R_0$, where $R_0$ is an initial reservoir density operator. At later times correlations between $S$ and $R$ will arise, but we assume that such coupling is very weak. If we write
$$\tilde{\chi}(t)=\tilde{\rho}(t)R_0+O(H_{SR}),$$
the {\bf Born approximation} consists of neglecting terms higher than second order in $H_{SR}$, so we can write
\begin{equation*}
\dot{\tilde{\chi}}=-\frac{1}{\hbar^2}\int_0^t [\tilde{H}_{SR}(t),[\tilde{H}_{SR}(t'),\tilde{\rho}(t')R_0]]\;dt'.
\end{equation*}
Finally, by replacing $\tilde{\rho}(t')$ with $\tilde{\rho}(t)$ we obtain the {\bf Born-Markov approximation}:
\begin{equation*}
\dot{\tilde{\chi}}=-\frac{1}{\hbar^2}\int_0^t [\tilde{H}_{SR}(t),[\tilde{H}_{SR}(t'),\tilde{\rho}(t)R_0]]\;dt'.
\end{equation*}

\end{document}